\documentclass[12pt]{article}

\usepackage{amsfonts}
\usepackage{amssymb}
\usepackage{amsthm}
\usepackage{amsmath}
\usepackage{amscd}
\usepackage{mathrsfs}
\usepackage{enumerate}
\usepackage{ulem}
\usepackage{fontenc}
\usepackage[all]{xy}
\usepackage{array}
\usepackage[english]{babel}

\def\CC {{\mathbb C}}     
\def\II {{\mathbb I}}     
\def\NN {{\mathbb N}}     
\def\PP {{\mathbb P}}     
\def\XX {{\mathbb X}}     
\def\YY {{\mathbb Y}}     

\def\lo  {\longmapsto}
\def\Lw  {\Longrightarrow}
\def\lw  {\longrightarrow}
\def\mc {\mathcal}

\def\ol  {\overline}

\def\rw  {\rightarrow}

\def\ul  {\underline}

\newtheorem{theorem}{Theorem}[section]
\newtheorem{lemma}[theorem]{Lemma}

\newtheorem{prop}[theorem]{Proposition}

\newtheorem{coro}[theorem]{Corollary}
\newtheorem{rem}{Remark}[section]
\newtheorem{df}{Definition}[section]

\begin{document}
\global\long\global\long\global\long\def\bra#1{\mbox{\ensuremath{\langle#1|}}}
\global\long\global\long\global\long\def\ket#1{\mbox{\ensuremath{|#1\rangle}}}
\global\long\global\long\global\long\def\bk#1#2{\mbox{\ensuremath{\ensuremath{\langle#1|#2\rangle}}}}
\global\long\global\long\global\long\def\kb#1#2{\mbox{\ensuremath{\ensuremath{\ensuremath{|#1\rangle\!\langle#2|}}}}}

\title{A link between Quantum Entanglement, Secant varieties
and Sphericity}

\author{Adam Sawicki$^{1,\,2}$ and Valdemar V. Tsanov$^{3}$
\\ \\
$^1$School of Mathematics, University of Bristol, \\
University Walk, Bristol BS8 1TW, UK \\ \\
$^2$Center for Theoretical Physics, Polish Academy of Sciences\\
Al. Lotnik\'ow 32/46, 02-668 Warszawa, Poland \\ \\
$^3$Fakult\"at f\"ur Mathematik, Ruhr-Universit\"at Bochum,\\
D-44780 Bochum, Germany }

\maketitle

\abstract{In this paper, we shed light on relations between three concepts studied in representations theory, algebraic geometry and quantum information theory. First - spherical actions of reductive groups on projective spaces. Second - secant varieties of homogeneous projective varieties, and the related notions of rank and border rank. Third - quantum entanglement. Our main result concerns the relation between the problem of the state reconstruction from its reduced one-particle density matrices and the minimal number of separable summands in its decomposition. More precisely, we show that sphericity implies that states of a given rank cannot be approximated by states of a lower rank. We call states for which such approximation is possible \textit{exceptional states}. For three, important from quantum entanglement perspective cases of distinguishable, fermionic and bosonic particles, we also show that non-sphericity implies the existence of exceptional states. Remarkably, the exceptional states
belong to non-bipartite entanglement classes. In particular, we show that the $W$-type states and their appropriate modifications are exceptional states stemming from the second secant variety for three cases above. We point out that the existence of the exceptional states is a physical obstruction for deciding the local unitary equivalence of states by means of the one-particle reduced density matrices. Finally, for a number of systems of distinguishable particles with known orbit structure we list all exceptional states and discuss their possible importance in entanglement theory. }

\section{Introduction}

We consider a projective algebraic variety given as $\XX\subset\PP(\mc H)$, where ${\mc H}$ is the state space of a composite system and $\XX$ consists of the coherent states of the system. The systems we consider are:

\begin{enumerate}
	\item $L$ distinguishable particles with the state space $\mc H_D=\mc H_1\otimes ...\otimes \mc H_L$,
	\item $L$ bosons (symmetric indistinguishable particles) with the state space $\mc H_B=S^L(\mc H_1)$, where $\mc H_1$ is one-boson space,
	\item $L$ fermions (antisymmetric indistinguishable particles) with the state space $\mc H_F=\bigwedge^L\mc H_1$, where $\mc H_1$ is one-fermion space.
\end{enumerate}

Throughout the paper states will be understood as points in the projective space rather than vectors in the Hilbert space. Physically, it means that we neglect the norm and global phase of vectors. The respective varieties of coherent states, compact symmetry groups and their complexifications, which are local unitary and invertible SLOCC operations, respectively, are:
\begin{enumerate}
	\item The Segre embedding of $\PP(\mc H_1)\times...\times\PP(\mc H_L)$ into $\PP({\mc H}_D)$, $K_D=SU(N_1)\times \ldots \times SU(N_L)$, $G_D=K_D^\mathbb{C}=SL(\mc H_1)\times...\times SL(\mc H_L)$,
	\item The Veronese embedding of $\PP(\mc H_1)$ into $\PP({\mc H}_B)$, $K=SU(N)$, $G=K^\mathbb{C}= SL(\mc H_1)$,
	\item The Pl\"ucker embedding of the Grassmannian $Gr(L,\mc H_1)$ into $\PP({\mc H}_F)$, $K=SU(N)$, $G=K^\mathbb{C}= SL(\mc H_1)$.
\end{enumerate}

The local unitary operations represent the unitary operations which are performed on each particle separately, i.e. the local unitary dynamics. The SLOCC operations are more general and in addition allow classical communication. In all these cases $\XX$ spans $\PP(\mc H)$, i.e. every state can be written as a linear combination of states from $\XX$. The {\it rank} of a state is then defined as the minimal number of summands in such a linear combination. An obvious, but important property is that the rank is invariant under the symmetries of the system.

Throughout the paper, when convenient, we will use Dirac bra-ket notation. The basis of $\mathbb{C}^N$ will be denoted by $\{\ket{0},\ldots,\ket{N-1}\}$. The tensor $\ket{i}\otimes\ldots\otimes\ket{j}$ will typically be written as $\ket{i\ldots j}$.

The notion of rank and the related secant varieties, which are introduced here in section \ref{Sec Rank&Sec}, have been studied for a long time and from various perspectives - algebraic geometry, representation theory, computational complexity, algebraic statistics, etc., see e.g. \cite{AOP-2009,Bucz-Lands-2011,CGG-2002,Landsberg-2012-book} and the references therein. Much is known and much is still a mystery. We intend to show that these notions are also relevant in the study of quantum entanglement from the physics side and spherical varieties from the mathematical side.

The phenomenon which we focus on is the following. In some cases, states of a given rank can be approximated by states of lower rank. We will call such states {\it exceptional}, as to distinguish them from the prototypical case: the matrices, where the set of matrices whose rank doesn't exceed a given $r$ is a closed set, so approximation to higher rank is impossible. We show that exceptional states occurs for most systems of the above mentioned types. In fact, they do not occur if and only if $L=2$. Furthermore, the spaces $\mc H_1\otimes \mc H_2$, $S^2(\mc H_1)$, $\bigwedge^2\mc H_1$ are exactly the ones, where the complex symmetry group $G$ acts spherically on $\PP ({\mc H})$. In fact, for the tensor representations we are interested in, we observe that the absence of exceptional states is equivalent to sphericity of the group action. The central result is formulated as Theorem \ref{Theo Spher<->no except}. Theorem \ref{Theo Spher<->no except} provides a new definition of sphericity for the considered
systems. Thus, the rather non-physical notion of a dense Borel group orbit is changed to a state-type obstruction.

Remarkably, the exceptional states we find, in all three cases, belong to non-bipartite SLOCC entanglement classes. More precisely, we show that $W$-type states are exceptional. For example for the system of three qubits $W=\frac{1}{3}(\ket{001}+\ket{010}+\ket{001})$, is exceptional. Combining the results of the current paper with the recently established connection between sphericity and local unitary equivalence of states \cite{SawHuckKus2012} we conclude in Theorem \ref{lu-result} that the existence of the exceptional states can be regarded as a physical, state-type obstruction for deciding the local unitary equivalence of states by means of the one-particle reduced density matrices.

From point of view of representation theory, it is natural to ask whether the equivalence between sphericity and absence of exceptional states persists in a more general situation. The spherical actions of reductive groups on projective spaces are completely classified, see \cite{Knop-MultFree-1998}. A case by case analysis using the list of \cite{Knop-MultFree-1998} shows that sphericity implies absence of exceptional states. This implication can also be deduced from a result of \cite{Bucz-Lands-2011}, where it is shown that rank and border rank coincide for a class of homogeneous projective varieties (subcominuscule varieties); the latter class is exactly the class obtained from spherical representations, taking account of the fact that sometimes several spherical representations give rise to the same variety of coherent states. However, the converse implication does not hold in general, i.e. there are irreducible representations of reductive groups, where exceptional states do not occur, and the action on
the projective space is not spherical. Perhaps the simplest example is given by the adjoint representation of $SL_3\CC$. The conclusion is that we should regard the equivalence between sphericity and lack of exceptional states as a phenomenon present in the setting of quantum entanglement, but not beyond.

The fact that exceptional states do occur for systems with more than two particles can be qualified, from the point of view of algebraic geometry, as folklore. All necessary results can be found in \cite{Landsberg-2012-book}. The value of this work, as we see it, is not so much in the proofs, but rather in exhibiting a connection between three separately well studies notions: rank, entanglement types and sphericity. Such a connection could, hopefully, lead to interactions and new developments.

The paper is organized as follows. We start with two simple motivating examples. The first describes a situation where exceptional states are absent and the second - where they are present. Next, in section \ref{Sec Rank&Sec} we give a rigorous definitions of rank, border rank and secant varieties and discuss some of their basic properties. In section \ref{Sec:A_short_overview} we present some known results about secant varieties for distinguishable particles, bosons and fermions. The sections 6 and 7 consist of the formulation and the proof of the main results of the paper. In section 8 we give a list of all exceptional states which appear for systems with the explicitly known orbit structure.

\section{Motivation - a simple example}

Before we proceed with the general definitions and results, we present two examples: one where exceptional states do not occur - the case of two qubits, and one where they do occur - three qubits. These systems are well studied from many aspects. Remarkably, in the theory of quantum entanglement, the exceptional state for three qubits is so-called W state.

\paragraph{Two qubits}

Consider two complex vector spaces of dimension 2, $\mc H_1\cong \mc H_2\cong \CC^2$. The tensor product ${\mc H}_D=\mc H_1\otimes \mc H_2$, or more precisely, its projective space $\PP=\PP(\mc H_D)$ is the state space for a system of two qubits. This state space carries an action of the compact group $K_D=SU(2)\times SU(2)$ and also of its complexification $G_D=SL_2(\CC)\times SL_2(\CC)$.  We concentrate on the action of $G_D$.

There are exactly two $G_D$-orbits in $\PP$. The first consists of all simple tensors (separable states), i.e. tensors which can be written\footnote{We denote by $[\psi]$ the point in $\PP(\mc H)$ corresponding to a vector $\psi\in{\mc H}\setminus\{0\}$.} as $[v\otimes v']$ with $v\in \mc H_1$, $v'\in \mc H_2$. The second orbit consists of all non-simple tensors. It turns out that any state $[\psi]$ in the second orbit can be written as a sum of two simple tensors. More precisely, there exist bases $v_1,v_2$ of $\mc H_1$ and $v_1',v_2'$ of $\mc H_2$ so that $\psi=v_1\otimes v_1' + v_2\otimes v_2'$. Let us denote the first orbit by $\XX$, so that the second one is $\PP\setminus \XX$. The set $\XX$ is actually a well know algebraic variety, called the Segre embedding\footnote{Incidentally, $\XX$ is also a quadric hypersurface, but we regard it here as a Segre embedding, since this is what we want to consider in a more general situation.} of $\PP^1\times \PP^1$.

The basic notion considered in this paper is the notion of {\it rank} of states. In this terminology, the states of the first orbit will be said to have rank 1 and the states of the second orbit to have rank 2, according to the minimal number of simple tensors necessary to express a given state. This terminology has a classical origin: the space $\mc H_1\otimes \mc H_2$ can be interpreted as the space of $2\times2$ matrices with complex entries and we recover the standard notion of rank of a matrix.

\paragraph{Three qubits}\label{Sec Motiv 3 qubits}
Let $\mc H_D=\mc H_1\otimes \mc H_2\otimes \mc H_3$, with $\mc H_j\cong \CC^2$. Let $\PP=\PP(\mc H_D)$. Let $G=SL(\mc H_1)\times SL(\mc H_2)\times SL(\mc H_3)$ act naturally on $\PP$. Let $\XX\subset\PP$ denote the variety of simple tensors; this is the Segre embedding of $\PP^1\times\PP^1\times\PP^1$ into $\PP$. Let $e_1,e_2$ be a basis of $\CC^2$. Then the orbits of $G$ in $\PP$ are the following:
\begin{gather*}
\begin{array}{ll}
\XX = \XX_1 & = G[e_1\otimes e_1\otimes e_1] \\
\XX_2^{\rm I} & = G[e_1\otimes e_1\otimes e_1 + e_2\otimes e_2\otimes e_2] \\
\XX_2^{\rm II} & = G[e_1\otimes e_1\otimes e_1 + e_1\otimes e_2\otimes e_2] \\
\XX_2^{\rm III} & = G[e_1\otimes e_1\otimes e_1 + e_2\otimes e_1\otimes e_2] \\
\XX_2^{\rm IV} & = G[e_1\otimes e_1\otimes e_1 + e_2\otimes e_2\otimes e_1] \\
\XX_3 & = G[e_1\otimes e_1\otimes e_2 + e_1\otimes e_2\otimes e_1 + e_2\otimes e_1\otimes e_1] \\
\end{array}
\end{gather*}
The lower index indicates rank. The upper index for $\XX_2$ distinguishes the four different orbits of rank two; observe that the cases II,III,IV differ only by permutation of the indices; these are so-called bi-separable states. It is not hard to see that the orbit $\XX_2^{\rm I}$ is open and dense in $\PP$. Hence all states can be approximated by states from $\XX_2^{\rm I}$. In particular, this is true for the state
$$
W = [e_1\otimes e_1\otimes e_2 + e_1\otimes e_2\otimes e_1 + e_2\otimes e_1\otimes e_1] \;.
$$
We say that $W$ has border-rank 2. However, it can be checked directly that $W$ cannot be written as a sum of two simple tensors, so its rank is equal to 3. Thus $W$ is an exceptional state - it can be approximated by states of lower rank. Remarkably, at the same time, the states $e_1\otimes e_1\otimes e_1 + e_2\otimes e_2\otimes e_2$ and  W are known to represent the two genuine entanglement classes of three qubits.

\section{Rank of states and secant varieties: general definitions}\label{Sec Rank&Sec}

After the simple motivating examples given above, we proceed with the rigorous definitions of the central notions of this article: rank, border rank and secant varieties. We also state some of their basic properties. A detailed treatment can be found in \cite{Landsberg-2012-book}, along with an account of the state of affairs for the cases concerning us here, from point of view of representation theory and algebraic geometry.

Let ${\mc H}$ denote a complex vector space of dimension $N$ and $\PP=\PP(\mc H)$ denote the associated projective space. If $v\in{\mc H}$ is a nonzero vector, we shall denote by $[v]$ its image in $\PP$.

\paragraph{Rank and maximal rank} \label{subsec:Rank_and_maximal_rank}

Let $\XX\subset\PP$ be an algebraic variety and let $\hat{\XX}\subset{\mc H}$ denote the affine cone over $\XX$. Then $\mathrm{span}\hat{\XX}$ is a linear subspace of ${\mc H}$. We denote by $\PP_\XX=\PP(\mathrm{span}\hat{\XX})$ the corresponding projective subspace of $\PP$. We say that $\XX$ spans $\PP$ if $\PP_\XX=\PP$; this is equivalent to the requirement that $\hat{\XX}$ contains a basis of ${\mc H}$. Assume that this is the case. Then every point in ${\mc H}$ can be written as a linear combination of points in $\hat{\XX}$. This allows us to define the following notion of rank of a vector with respect to $\XX$, for non-zero $\psi\in{\mc H}$,
\begin{equation}
{\rm rk}[\psi] = {\rm rk}_\XX[\psi] = \min \{r\in\NN : \psi={x}_1+\dots +{x}_r \;{\rm with}\; [x_j]\in\XX\} \;.
\end{equation}
\noindent The sets
\begin{equation}\label{rank-subset}
\XX_r = \{ [\psi]\in\PP : {\rm rk}[\psi] = r \}\,,\quad r = 1,2,...,
\end{equation}
\noindent are called the rank subsets of $\PP$ with respect to $\XX$. Since $\XX$ spans $\PP$, we have $\XX_r=\emptyset$ for $r>N$. In the following proposition we state some basic properties of the rank subsets. We skip the formal proof, since all properties follow immediately from the definitions.

\begin{prop}\label{Prop Properties X_r}
The following properties hold:

{\rm (i)} $\XX_1=\XX$.

{\rm (ii)} There exists $r_m\in\{1,...,N\}$, such that $\XX_{r_m}\ne \emptyset$ and $\XX_{r}=\emptyset$ for $r>r_m$.\\
\indent\indent\, The number $r_m$ is called the {\rm maximal rank} of $\PP$ with respect to $\XX$.

{\rm (iii)} If $r\in\{1,...,r_m\}$, then $\XX_{r}\ne \emptyset$.

{\rm (iv)} The projective space $\PP$ decomposes as a disjoint union $\PP=\XX_1\sqcup \dots \sqcup \XX_{r_m}$.
\end{prop}

\paragraph{Secant varieties, typical rank}

Let $r\in\{2,...,r_m\}$ and $\XX_r\subset\PP$ be as defined in (\ref{rank-subset}). Let us first note that the subset $\XX_r\subset\PP$ is not closed. It can be easily seen, because we have $\XX\subset\ol{\XX_r}$ and $\XX\nsubseteq\XX_r$\footnote{Here and in what follows we use $\ol{S}$ to denote the Zariski closure of a subset $S\subset\PP$.}. The $r$-th secant variety of $\XX$ is defined as
\begin{equation}
\sigma_r(\XX) = \ol{\bigcup\limits_{x_1,...,x_r\in\XX} \PP_{x_1...x_r} } \;,
\end{equation}
where $\PP_{x_1...x_r}$ stands for the projective subspace of $\PP$ spanned by the points $x_1,...,x_r$. It can also be written as
\begin{equation}
\sigma_r(\XX) = \ol{\bigsqcup\limits_{s\leq r}\XX_s} \subset\PP \;.
\end{equation}
In the proposition below we state some easy to verify, basic properties of secant varieties.

\begin{prop} The following properties hold:

{\rm (i)} $\sigma_1(\XX)=\XX_1=\XX$.

{\rm (ii)} $\sigma_r(\XX)\subset\sigma_{r+1}(\XX)$.

{\rm (iii)} There exists a minimal $r_g\in\{1,...,r_m\}$ such that $\sigma_{r_g}(\XX)=\PP$ and $\sigma_{r_g-1}(\XX)\ne\PP$. \\
            \indent \indent \, This number is called the {\rm typical rank} of $\PP$ with respect to $\XX$.

{\rm (iv)} For $r\in\{1,...,r_g\}$ the rank subset $\XX_r$ is Zariski open and dense in $\sigma_r(\XX)$ and we\\
 \indent \indent \, have $\sigma_r(\XX)=\ol{\XX_r}$.
\end{prop}

\paragraph{Border rank and exceptional states}

As we have already seen for three qubits, it may happen that a state with a given rank can be approximated by states of strictly lower rank. Here we introduce the notion of a border rank which captures this kind of behavior. Let $[\psi]\in\PP$. The border rank of $[\psi]$ with respect to $\XX$ is defined as
$$
\ul{\rm rk}[\psi] = \ul{\rm rk}_\XX[\psi] = \min\{r\in\NN:[\psi]\in\ol{\XX_r}\} \;.
$$

\begin{prop} The following properties hold:

{\rm (i)} ${\rm rk}[\psi]\geq\ul{\rm rk}[\psi]$.

{\rm (ii)} $\ul{\rm rk}[\psi] = \min\{r\in\NN:[\psi]\in\sigma_r(\XX)\}$.
\end{prop}

\begin{df}
States $[\psi]\in\PP$, for which ${\rm rk}[\psi]\ne \ul{\rm rk}[\psi]$, are called {\rm exceptional}.
\end{df}

Clearly, if $[\psi]$ is exceptional, then $\ul{\rm rk}[\psi]<{\rm rk}[\psi]$, so exceptional states are states which can be approximated by states of lower rank.

\paragraph{Expected dimensions of secant varieties}

We can easily give a set of parameters sufficient to determine a point on $\sigma_r(\XX)$; this gives an upper bound on the dimension of the secant variety. Namely, a generic point in $\sigma_r(\XX)$ is obtained as follows. We take $r$ points on $\XX$, which gives $r\dim\XX$ parameters, and then, assuming that these points are linearly independent, we take a point in the $(r-1)$-dimensional projective space spanned on them. In total we get $r\dim\XX+(r-1)$ parameters. Since there are no obvious relations between these parameters, this number is called the {\it expected dimension} of $\sigma_r(\XX)$, when it does not exceed $N-1=\dim\PP$. It is denoted by
$$
{\rm e}\dim\sigma_r(\XX)=\min\{r\dim\XX+(r-1),N-1\} \;.
$$
If $\dim\sigma_r(\XX)\ne {\rm e}\dim\sigma_r(\XX)$, the secant variety is called defective and the difference between the two dimensions is called the {\it defect}.

With this notion in hand, we can calculate an expected value $r_{eg}$ for the typical rank of $\PP$ with respect to $\XX$, called the expected typical rank. This is the minimal $r$ for which $r\dim\XX+(r-1)\geq N-1$, so
$$
r_{eg} = \left\lceil \frac{N}{\dim\XX+1} \right\rceil \;.
$$
Clearly, $r_g\leq r_{eg}$.

\section{Distinguishable particles, bosons and fermions}\label{Sec Cases of interest}

In the following we show that the algebraic varieties of coherent states, for our three cases of interest, in fact span the respective projective spaces $\PP$. This in turn implies that for distinguishable particles, bosons and fermions there are a well defined notions of rank and border rank.

\paragraph{Distinguishable particles and the Segre variety}\label{Sec Def Distinguishable}

Let $\mc H_1,...,\mc H_L$ be complex vector spaces with finite dimensions $n_j=\dim \mc H_j$. Let $G_D=GL(\mc H_1)\times...\times GL(\mc H_L)$. Let ${\mc H}_D=\mc H_1\otimes...\otimes \mc H_L$ and $\PP=\PP(\mc H_D)$. Then $\dim{\mc H_D}=N=n_1\cdot ...\cdot n_L$ and $\dim\PP=N-1$. The group $G_D$ acts naturally on ${\mc H}
_D$ and this representation is irreducible. There is of course an induced action on $\PP$. Consider the map
\begin{gather*}
\begin{array}{cccc}
\mathrm{Segre} : & \PP(\mc H_1)\times...\times\PP(\mc H_L) & \lw & \PP \\
        & ([v_1], ...,[v_L]) & \lo & [v_1\otimes ...\otimes v_L] \;.
\end{array}
\end{gather*}
It is not hard to see that this map defines a $G_D$-equivariant embedding, called the Segre embedding of $\PP(\mc H_1)\times...\times\PP(\mc H_L)$. We denote the image of the Segre map by $\XX\subset\PP$. This is an algebraic variety, called the Segre variety, it consists of (the projectivizations of) all simple tensors. The Segre variety is the unique closed $G_D$-orbit in $\PP$. In physical terms, $\XX$ is the orbit of coherent or separable states.

The variety $\XX$ spans $\PP$. Indeed, if we fix bases for $\mc H_1,...,\mc H_L$, then the tensor products of the basis vectors form a basis of simple tensors for ${\mc H}_D$. Thus every tensor can be expressed as a linear combination of simple tensors and we have well defined notions of rank and border rank of states with respect to $\XX$, as defined in section \ref{Sec Rank&Sec}.

\paragraph{Bosons and the Veronese variety}

Let $\mc H_1$ be a complex vector space of dimension $n$. Let $G=GL(\mc H_1)$. Let ${\mc H_B}=S^L(\mc H_1)$, for some fixed positive integer $L$, and let $\PP=\PP(\mc H_B)$. Then $\dim {\mc H_B} = N = \binom{L+n-1}{L}$ and $\dim\PP=N-1$. The group $G$ acts naturally on ${\mc H_B}$ and this representation is irreducible. There is an induced action on $\PP$. Consider the map
\begin{gather*}
\begin{array}{cccc}
\mathrm{Ver}_L : & \PP(\mc H_1) & \lw & \PP \\
        & [v] & \lo & [v^L] \;.
\end{array}
\end{gather*}
This map is an embedding, called the $L$-th Veronese embedding of $\PP(\mc H_1)$. Let us denote the image by $\XX\subset\PP$. This is an algebraic subvariety of $\PP$ called the veronese Variety. It consists of all symmetric $L$-tensors over $\mc H_1$ which are powers of 1-tensors. Since $G$ acts transitively on $\PP(\mc H_1)$ and the map $\mathrm{Ver}_L$ is $G$-equivariant, $\XX$ is a single closed $G$-orbit. Since the representation of $G$ on ${\mc H_B}$ is irreducible, $\XX$ is the only closed $G$-orbit in $\PP$.

The variety $\XX$ spans $\PP$. This is, perhaps, not as obvious as in the case of the Segre variety, because the standard basis for $S^L(\mc H_1)$ is the monomial basis obtained from a given basis $\{e_1,...,e_n\}$ of $\mc H_1$. The only monomials which belong to $\XX$ are $[e_j^L]$. The other monomials, like $[e_1e_2^{L-1}]$ do not belong to $\XX$. To see that $\XX$ actually spans $\PP$, observe that, since $\XX$ is a closed $G$-orbit, $\mathrm{span}\XX$ gives a subspace of ${\mc H_B}$ preserved by $G$, i.e. a subrepresentation. But ${\mc H}$ is irreducible, so we must have $\mathrm{span}\XX=\PP$. Hence we have well defined notions of rank and border rank in $\PP$ with respect to $\XX$.

\paragraph{Fermions and the Grassmann variety}\label{Sec Def Anticommuting}

Let $\mc H_1$ be a complex vector space of dimension $n$ and $G=GL(\mc H_1)$. Let $Gr(L,\mc H_1)$ denote the $L$-th Grassmann variety, for some fixed integer $L$ with $1\leq L\leq n$; this is the variety of all $L$-dimensional linear subspaces of $\mc H_1$. Then $Gr(L,\mc H_1)$ carries a natural $G$-action. Let ${\mc H_F}=\bigwedge^L\mc H_1$ and $\PP=\PP(\mc H_F)$. Then $\dim {\mc H_F}=N=\binom{n}{L}$ and $\dim\PP=N-1$. There is a natural linear representation of $G$ on ${\mc H_F}$, which is irreducible. Consider the map
\begin{gather*}
\begin{array}{cccl}
\mathrm{Pl}_L : & Gr(L,\mc H_1) & \lw & \PP \\
     & U & \lo & [u_1\wedge...\wedge u_L] \,,\quad \textrm{where $u_1,...,u_L$ is a basis of $U$} \,.
\end{array}
\end{gather*}
Here $U\subset \mc H_1$ is an $L$-dimensional subspace, i.e. an element of $Gr(L,\mc H_1)$. The image $\mathrm{Pl}_L(U)$ is well defined, i.e. does not depend on the choice of basis, because for two different bases $\{u_j\}$ and $\{u_j'\}$ the exterior products $u_1\wedge...\wedge u_L$ and $u_1'\wedge...\wedge u_L'$ differ only by a scalar, so in the projective space $\PP$ we have $[u_1\wedge...\wedge u_L]=[u_1'\wedge...\wedge u_L']$. The map $\mathrm{Pl}_L$ is a $G$-equivariant embedding, called the Pl\"ucker embedding of the Grassmann variety. Let us denote the image by $\XX=\mathrm{Pl}_L(Gr(L,\mc H_1))$. Then $\XX$ is the unique closed $G$-orbit in $\PP$, it consists of all decomposable antisymmetric tensors in ${\mc H_F}$. The variety $\XX$ spans $\PP$ because it contains the standard basis of ${\mc H_F}$ consisting of the exterior products of $L$-tuples of elements of a given basis of $\mc H_1$. Hence we have a well defined notions of rank and border rank in $\PP$ with respect to $\XX$.

\begin{rem}\label{Rem Wedge^k=Wedge^n-k}
Unless otherwise specified, we shall only consider the cases where $L\leq n/2$. This is sufficient for our purposes, because of an isomorphism $\bigwedge^L\mc H_1\cong\bigwedge^{n-L}\mc H_1$. Such an isomorphism is not canonical, in fact the two spaces are dual to each other. So, we can obtain an isomorphism by choosing a positive definite Hermitian form on $\mc H_1$. Then the mapping $Gr(L,\mc H_1)\rw Gr(L,\mc H_1)$, given by $U\mapsto U^\perp$, is an isomorphism which agrees with the respective Pl\"ucker embeddings. The Hermitian form cannot be $G$-equivariant, but the associated unitary group $K=U(\mc H_1)$ is a maximal compact subgroup of $G$. Thus, although the isomorphism $\bigwedge^L\mc H_1\cong\bigwedge^{n-L}\mc H_1$ is not $G$-equivariant, there is a 1-1 correspondence between the $G$-orbits. Also, since the isomorphism sends simple tensors to simple tensors, the rank and border rank behave in the same way.
\end{rem}

\section{A short overview of known results about secant varieties}\label{Sec:A_short_overview}

The literature on secant varieties is vast and diverse, and a real overview is a task we do not intend to pursue here. We only give a very small sample of theorems conserning the situations described in Section \ref{Sec Def Distinguishable}. Although the only result we will need is Theorem \ref{Theo rank of monom} we found it is useful to state some known results at least for the most interesting cases of quantum entanglement. The interested reader is referred to \cite{Landsberg-2012-book} for an extensive overview of the subject.

\subsubsection*{L-qubit system}

Assume $n_j=2$, i.e. $\mc H_j=\CC^2$, for all $j$. Thus ${\mc H_D}=(\CC^2)^{\otimes L}$ and the Segre variety $\XX$ is an embedding of the $L$-fold product $\PP^1\times\dots\times\PP^1$ into $\PP^{2^L-1}$. The dimensions of all secant varieties for this case have been determined in \cite{CGG-2011}. Below we state the main results from this paper.

\begin{theorem}\label{Theo SecVar Cubits} {\rm (\cite{CGG-2011})}
The dimension of the secant variety $\sigma_r(\XX)$ equals the expected dimension $rL+r-1$ in all cases but one. The exceptional case is $r=3,L=4$, i.e. four qubits and $\sigma_3$, where ${\rm dim}\sigma_3(\PP^1\times\PP^1\times\PP^1\times\PP^1) = 13$, while the expected dimension is $14$.
\end{theorem}

\begin{coro}\label{Coro SecVar Cubits}
The minimal number of simple tensors in $(\CC^2)^{\otimes L}$ necessary to express a generic tensor as a linear combination is the expected number $r_g=r_{eg}= \left\lceil \frac{2^L}{L+1} \right\rceil$.
\end{coro}

\subsubsection*{L-boson system}

The following theorem of Alexander and Hirschowitz determines which secant variety to the Veronese variety fills the ambient projective space.

\begin{theorem}{\rm (\cite{Alex-Hirsch})} A generic element of $S^L(\mc H_1)$ with $\mc H_1\cong\CC^n$, i.e. a generic state of $L$ $n$-level bosons, can be written as the sum of the expected number $r_{eg}=\lceil \frac{1}{n}\binom{n-1+L}{L} \rceil$, and no fewer, $L$-th powers of elements of $\mc H_1$, with the following exceptions: $(n,L)=(3,4),(4,4),(5,4),(5,3)$ and $(n,2)$ for all $n\geq2$.
\end{theorem}

\noindent The problem of determining the rank of a symmetric tenson with repsect to the Veronese vareity is known as Waring's problem. A result which plays a role in the next section concerns the rank of a monomial. Below, we state two theorems of Carlini, Catalisano and Geramita, solving Waring's problem for monomials and sums of coprime monomials. Let $e_1,...,e_n$ be a basis of $\mc H_1$. A monomial in $S^k(\mc H_1)$ is a product of powers of the basis elements, i.e. $M=e_1^{\alpha_1}...e_n^{\alpha_n}$ with $\alpha_1+\dots+\alpha_n=L$. The basis can be reordered so that the nonzero exponents are increasing and the zero exponents are at the end, i.e. $M=e_1^{\alpha_1}...e_m^{\alpha_m}$ with $1\leq \alpha_1\leq\dots\leq\alpha_m$ and $m\leq n$.

\begin{theorem}\label{Theo rank of monom}
{\rm (\cite{CarCatGer-2012-Waring})} Let $M\in S^L(\mc H_1)$ be a monomial as above. Then
$$
{\rm rk}(M) = \prod_{j=2}^{m} (\alpha_j+1) \;.
$$
\end{theorem}

\noindent Examples show that the rank of a monomial may be greater or smaller than the rank $r_g$ of a generic element of $S^L(\mc H_1)$.

\begin{theorem}{\rm (\cite{CarCatGer-2012-Waring})} Suppose $f\in S^L(\mc H_1)$ has the form $f=M_1+\dots + M_l$, where $M_1,...,M_l$ are pairwise coprime monomials of degree $L$. Then
$$
{\rm rk} (f) = \sum_{j=1}^{l} {\rm rk} (M_j) \;.
$$
Furthermore, if $M=M_1...M_l$ denotes the product, then
$$
{\rm rk}(M_j) \leq {\rm rk}(f) \leq {\rm rk}(M) \;.
$$
\end{theorem}

\subsubsection*{L-fermion system}

Results about the secant varieties of Grassmann varieties were obtained in \cite{CGG-2005}. In particular, the following theorem is proven therein.

\begin{theorem}
{\rm (\cite{CGG-2005})}
Let $\XX = Pl(Gr(L,\CC^n)) \subset \PP(\bigwedge^L\CC^n)$ be the image of the Pl\"ucker embedding of the Grassmann variety, with $L\leq n/2$ (see Section \ref{Sec Def Anticommuting} for details).

{\rm (i)} If $L=2$, then the typical rank is $r_g=\lfloor\frac{n}{2}\rfloor$. For $1<r<r_g$, the secant variety\\
 \indent \indent $\sigma_r(\XX)$ is defective with defect $e\dim\sigma_r(\XX)- \dim\sigma_r(\XX) = 2r(s-1)$.

{\rm (ii)} If $L\geq 3$ and $Lr\leq n$, then $\sigma_r(\XX)$ has the expected dimension $r\dim\XX+r-1$.
\end{theorem}

\section{The main results}\label{Sec:the_main_results}

Let $G$ be a reductive complex Lie group, i.e. $G$ equals the complexification $K^\CC$ of a compact Lie group $K$. Recall that a Borel subgroup $B\subset G$ is a maximal solvable subgroup and all Borel subgroups of $G$ are conjugate. In the case when $G$ is $SL_N(\CC)$, an example of a Borel subgroup is given by all upper-triangular matrices in $G$. In case $G$ splits as a direct product $G= G_1\times\dots\times G_L$, any Borel subgroup of $G$ is obtained as a product $B=B_1\times\dots\times B_k$, where $B_j$ is a Borel subgroup of $G_j$.

Suppose that $G$ acts on an algebraic (affine or projective) variety $\YY$. We say that $\YY$ is a $G$-variety. The variety $\YY$ is called a spherical $G$-variety, if a Borel subgroup $B$ of $G$ has an open dense orbit in $\YY$.

Let $\YY$ be a spherical $G$-variety. Since some, and therefor any, Borel subgroup $B$ of $G$ has an open dense orbit in $\YY$, it follows that $G$ also has an open dense orbit in $\YY$. Let $y\in \YY$ be a point from this open dense $G$-orbit and $G_y$ be the isotropy group at this point. Then the open $G$-orbit in $\YY$ is isomorphic to the coset space $G/G_y$. The open $B$-orbit is necessarily contained in the open $G$-orbit. Thus $B$ has an open orbit in $G/G_y$.

We are mostly concerned with a very specific type of spherical varieties described as follows. Let $\rho:G\to GL({\mc H})$ be a linear representation of a reductive group $G$. If the resulting action of $G$ on $\mc H$ is spherical, we say that $\rho$ is a {\it spherical representation}. In such a case, it is easy to see that the induced action of $G$ on the projective space $\PP=\PP(\mc H)$ is also spherical. The converse is not automatically true, but the problem can easily be removed as follows. Suppose the $G$ acts spherically on $\PP(\mc H)$ via a linear representation $\rho$ as above. Let $\tilde{G}=\CC^\times\times G$ and consider the representation
$$
\tilde{\rho}:\tilde{G} \lw GL(\mc H) \,,\quad \rho(\lambda,g)=\lambda\rho(g) \;.
$$
This representation defines a spherical action of $\tilde{G}$ on ${\mc H}$.

Suppose now $\rho:G\lw GL(\mc H)$ is a spherical representation which is in addition irreducible. Then $G$ has a unique closed orbit $\XX\subset\PP$. We may consider the secant varieties $\sigma_r(\XX)$ of $\XX$ in $\PP$. The irreducibility implies that $\XX$ spans $\PP$. Thus we have well defined rank and border rank functions on $\PP$ with respect to $\XX$. The main theorem of the paper, which we prove in the next section, is the following

\begin{theorem}\label{Theo Spher<->no except}
Suppose that we have one of the following three configurations\footnote{See Section \ref{Sec Cases of interest} for details.} of a state space ${\mc H}$, a complex reductive Lie group $G$ acting irreducibly on ${\mc H}$, and a variety of coherent states $\XX\subset\PP(\mc H)$, which is the unique closed $G$-orbit in the projective space $\PP(\mc H)$.

{\rm (i)} ${\mc H_D}=\mc H_1\otimes ...\otimes \mc H_L$, $G_D=GL(\mc H_1)\times ...\times GL(\mc H_L)$, $\XX=\mathrm{Segre}(\PP(\mc H_1)\times...\times\PP(\mc H_L))$.

{\rm (ii)} ${\mc H_B}=S^L(\mc H_1)$, $G=GL(\mc H_1)$, $\XX=\rm{Ver}_L(\PP(\mc H_1))$.

{\rm (iii)} ${\mc H_F}=\bigwedge^L\mc H_1$, $G=GL(\mc H_1)$, $\XX=\rm{Pl}(Gr(L,\mc H_1))$.

\noindent Then the action of $G$ on ${\PP(\mc H_{B,F}})$ (resp. $G_D$ on $\PP(\mc H_{D}$)) is spherical if and only if there do not exist exceptional states in $\PP(\mc H_{B,F})$ (resp. $\PP(\mc H_{D})$) with respect to $\XX$. In other words, sphericity of the representation is equivalent to the property that states of any given rank cannot be approximated by states of lower rank.
\end{theorem}
\noindent Note that theorem \ref{Theo Spher<->no except} provides a new definition of sphericity for the considered systems. The notion of rather non-physical dense Borel group orbit has been changed to a state-type obstruction. In the next section, we will see that some exceptional states belong to non-bipartite SLOCC entanglement classes. In particular, $W$-type state are exceptional. Hence the obstruction to sphericity is not only state-type but also entanglement-type.

Recently the importance of spherical varieties for the problem of local unitary equivalence of states has been pointed out by the authors of \cite{SawHuckKus2012}. Two states of $L$ distinguishable particles, $L$ bosons or $L$ fermions are called local unitary equivalent if and only if they can be connected by the action of $K_D=SU(N_1)\times\ldots \times SU(N_L)$, $K=SU(N)$ or $K=SU(N)$, respectively. It is well known that complex projective spaces are K\"{a}hler and hence symplectic manifolds. The action of $K_D$ on $\mathbb{P}(\mathcal{H}_D)$ and action of $K$ on $\mathbb{P}(\mathcal{H}_{B,F})$ are hamiltonian and therefore there exist the momentum maps $\mu:\mathbb{P}(\mathcal{H}_D)\rightarrow \mathfrak{k}^\ast_D$ and $\mu:\mathbb{P}(\mathcal{H}_{B,F})\rightarrow \mathfrak{k}^\ast$. Using the customary identification of Lie algebra with its dual one can show that \cite{SawHuckKus-2011}

\begin{eqnarray}
\mu([v])=\frac{i}{2}\{\widetilde{\rho}_1([v]),\ldots,\widetilde{\rho}_L([v])\}\,\,\, \mathrm{for\,distinguishable\,particles},\\
\mu([v])=\frac{i}{2}\widetilde{\rho}_1([v])\,\,\, \mathrm{for\,bosons\,and\,fermions},
\end{eqnarray}
where in all three cases $\widetilde{\rho}_i([v])=\rho-\frac{1}{n_i}I_{n_i}$ and $\rho_i([v])$ are the one-particle reduced density matrices of a state $[v]$. The natural question stemming from the local unitary (LU) equivalence problem is: when could one decide about LU equivalence by means of reduced density matrices only? The answer to this question is provided by Brion's theorem, \cite{brion87}, stating that an irreducible $K^\CC$-variety is spherical if and only if the momentum map distinguishes $K$-orbits. The formulation of this theorem adjusted to our setting reads:

\begin{theorem}\label{brion result}Let $K$ be a connected
compact Lie group acting on $\PP (\mathcal{H})$ by a Hamiltonian action and let $G=K^{\mathbb{C}}$.
The following are equivalent
\begin{enumerate}
\item $\PP (\mathcal{H})$ is a spherical variety.
\item Two states $[v_1]$ and $[v_2]$ are local unitary equivalent if and only if the spectra of their reduced one-particle density coincide.
\end{enumerate}
\end{theorem}	

\noindent Combining theorem \ref{Theo Spher<->no except} with the above theorem we get
\begin{theorem}\label{lu-result}
The existence of the exceptional states is a physical, state-type obstruction for deciding the local unitary equivalence by means of the one-particle reduced density matrices.
\end{theorem}

\begin{rem}
As we mentioned in section \ref{Sec Motiv 3 qubits}, for three qubits the only exceptional states are $W$-type states which have rank three and border rank two. The states of rank two which can approximate the state $W$ belong to $\XX_2^{\rm I}$, which is the SLOCC class of $GHZ$ state. Note that if we remove this class from $\mathbb{P}(\mathcal{H})$ the remaining states, i.e. the closure of the SLOCC class of $W$, are non-exceptional. On the other hand, in \cite{SWK13} it was shown that the closure of the SLOCC class of $W$ is spherical, i.e. it is possible to decide the local unitary equivalence by means of the one-particle reduced density matrices for $W$-type states. In the light of theorem \ref{lu-result} one can argue that it is because we removed from $\mathbb{P}(\mathcal{H})$ the states which were responsible for appearance of exceptional states. The generalization of this kind of reasoning to other systems is not immediate, however, it might give an insight to the classification of exceptional states.
\end{rem}

\section{The proof}\label{Sec Spher->no except}

In this section we present a proof of the main results stated in section \ref{Sec:the_main_results}. The proof is organized as follows. First, in theorem \ref{theol->}, we show that sphericity implies lack of exceptional states. Next, we proceed with a proof of the opposite implication for $L$ distinguishable particles, $L$ bosons and $L$ fermions. It essentially consists of two parts. The first one reduces the problem to some particular low dimensional cases (see section \ref{Sec Reduction}). Then we prove the result for these cases, which finishes the proof of Theorem \ref{Theo Spher<->no except}.

\begin{theorem}\label{theol->} Let $G\to GL(\mc H)$ be an irreducible representation of a reductive complex Lei group $G$, such that the action of $G$ on $\PP(\mc H)$ is spherical. Let $\XX\subset \PP(\mc H)$ be the closed $G$-orbit. Then rank and border rank on $\PP(\mc H)$ with respect to $\XX$ coincide, i.e.
$$
{\rm rk}_\XX[\psi] = \ul{\rm rk}_\XX[\psi] \;,
$$
for all $[\psi]\in\PP$. In other words, there are no exceptional states in $\PP$.
\end{theorem}

\begin{proof} All spherical representations of reductive groups are classified; see Knop \cite{Knop-MultFree-1998}. The theorem can be proved case by case using Knop's list. In particular, for the systems of our interest the sphericity is present only when $L=2$ (see \cite{SawHuckKus2012} for a detailed discussion). So we need to consider arbitrary $2$-tensors for distinguishable particles and symmetric or antisymmetric $2$-tensors for bosons an fermions. Note, however, that any such tensor $\psi$ can be represented by a matrix $M$, which is arbitrary, symmetric or antisymmetric, respectively. We have ${\rm rk}_\XX[\phi]={\rm rk}M$ in the cases of distinguishable particles and symmetric tensors, and ${\rm rk}_\XX[\psi]=\frac{1}{2}{\rm rk}M$ in the case of fermions. The set of matrices of rank smaller or equal to a given $r$ is a closed set. Hence the rank and the border rank coincides.
\end{proof}

\subsection{Absence of exceptional states implies sphericity}

In this section we prove the following result.

\begin{theorem} Suppose that we have one of the following three configurations\footnote{See Section \ref{Sec Cases of interest} for details.} of a state space ${\mc H}$, a complex reductive Lie group $G$ acting irreducibly on ${\mc H}$, and a variety of coherent states $\XX\subset\PP(\mc H)$, which is the unique closed $G$-orbit in the projective space $\PP(\mc H)$.

{\rm (i)} ${\mc H}=\mc H_1\otimes ...\otimes \mc H_L$, $G=GL(\mc H_1)\times ...\times GL(\mc H_L)$, $\XX=\mathrm{Segre}(\PP(\mc H_1)\times...\times\PP(\mc H_L))$, with $n_j=\dim V_j\geq 2$.

{\rm (ii)} ${\mc H}=S^L(\mc H_1)$, $G=GL(\mc H_1)$, $\XX=\mathrm{Ver}_L(\PP(\mc H_1))$, with $n=\dim \mc H_1\geq 2$.

{\rm (iii)} ${\mc H}=\bigwedge^L\mc H_1$, $G=GL(\mc H_1)$, $\XX=\mathrm{Pl}_L(Gr(L,\mc H_1))$, with $n=\dim \mc H_1\geq 6$.

Suppose that $L\geq 3$ and, only in the case {\rm (iii)}, also $n-L\leq 3$. Then there exist exceptional states in $\PP(\mc H)$. In other words, there exist states in $\PP(\mc H)$ which can be approximated by states of lower rank. Moreover, exceptional states already occur in the second secant variety $\sigma_2(\XX)$.
\end{theorem}

\begin{rem} The assumption $L\geq 3$ and, in the case (iii), $n-L\leq 3$, is equivalent to the assumption that the action of $G$ on ${\mc H}$ is not spherical. Thus, the above theorem can be phrased as: non-sphericity implies presence of exceptional states. Combined with theorem \ref{theol->} this implies that, for the cases of interest, sphericity is equivalent to the lack of exceptional states, thus proving theorem \ref{Theo Spher<->no except}.
\end{rem}

\begin{proof} We proceed by finding some basic examples of exceptional states in the smallest non-spherical representations and then reducing calculation of rank in the general case to these small-dimensional cases. The proof of the reduction lemma and the details of the examples are deferred to the subsequent subsections.

{\it Case} (i): Let $U_1,U_2,U_3$ be 2-dimensional subspaces of $\mc H_1,\mc H_2,\mc H_3$, respectively. Let $U_j\subset \mc H_j$ be a 1-dimensional subspace for $j=4,...,L$. Let $\tilde{\mc H} = U_1\otimes...\otimes U_L\subset {\mc H}$. Then we fall in the situation described in Section \ref{Sec Reduction} and we shall use the notation introduced therein, in particular $\tilde{\XX}=\XX\cap\PP(\tilde{\mc H})$. So, by Lemma \ref{Lemma Reduction}, for states $[\psi]\in\PP(\tilde{\mc H})$ we have ${\rm rk}_\XX[\psi]={\rm rk}_{\tilde{\XX}}[\psi]$. Now, notice that $\tilde{\mc H} = U_1\otimes U_2\otimes U_3$, because the 1-dimensional factors can be dropped out. We know, from the example given in section \ref{Sec Example 3 qubits} (or from Section \ref{Sec Motiv 3 qubits}), that the $W$-state has rank 3 but can be approximated by states of rank 2. Thus the $W$-state is exceptional in $\tilde{\mc H}$ and, consequently, in ${\mc H}$.

{\it Case} (ii): Here we can directly apply the result of the example computed in Section \ref{Sec Example Bosons}.

{\it Case} (iii): Let $U\subset \mc H_1$ be any subspace of dimension $L+3$. Let $\tilde{\mc H}=\bigwedge^LU$. Then $\tilde{\mc H}$ is naturally included as a subspace of ${\mc H}$ and we fall again in the situation treated in Section \ref{Sec Reduction}. Put $\tilde{\XX}=\XX\cap\PP(\tilde{\mc H})$. We can apply Lemma \ref{Lemma Reduction} and so we know that for $[\psi]\in\PP(\tilde{\mc H})$ we have ${\rm rk}_\XX[\psi]={\rm rk}_{\tilde{\XX}}[\psi]$. Thus it is sufficient to find exceptional states in $\tilde{\mc H}$. Notice that, since $\dim U=L+3$, we have $\bigwedge^LU\cong \bigwedge^3U$ and the rank functions agree (see Remark \ref{Rem Wedge^k=Wedge^n-k}). Since $\dim U\geq6$, we can choose a 6-dimensional subspace $Y\subset U$. Put $\bar{\mc H}=\bigwedge^3 Y$. Then $\bar{\mc H}$ is a subspace of $\bigwedge^3U$ to which we can apply once again the reduction procedure from Section \ref{Sec Reduction}. It is thus sufficient to find exceptional states in $\bigwedge^3\CC^6$. This is done in Section \ref{Sec
Example 3 fermions}.
\end{proof}

\subsubsection{Reduction to lower dimensional cases}\label{Sec Reduction}

In this section, we consider particular cases of a state space ${\mc H}$ and a subspace $\tilde{\mc H}\subset{\mc H}$. There is a natural embedding $\PP(\tilde{\mc H})\subset\PP(\mc H)$. We have the variety of coherent states $\XX\subset\PP(\mc H)$ for the first system. The variety of coherent states for the subsystem is of the form $\tilde{\XX}=\XX\cap\PP(\tilde{\mc H})$. Now for any state $[\psi]\subset\PP(\tilde{\mc H})$ we have two notions of rank: with respect to $\XX$ or $\tilde{\XX}$. We shall prove that, in our three cases of interest, these ranks agree, i.e. ${\rm rk}_\XX[\psi]={\rm rk}_{\tilde{\XX}}[\psi]$. This will enable us to compute rank for some basic low-dimensional systems and deduce results for much more general situations. We start by defining explicitly the spaces ${\mc H}$ and $\tilde{\mc H}$ we want to consider for distinguishable particles, bosons and fermions.\\

{\it Case 1: Distinguishable particles.} Let $\mc H_1,...,\mc H_L$ be finite dimensional complex vector spaces and ${\mc H}=\mc H_1\otimes...\otimes \mc H_L$. Let $U_j\subset \mc H_j$ be a subspace, for $j=1,...,L$, and $\tilde{\mc H}=U_1\otimes...\otimes U_L$. Then we have natural embeddings $\tilde{\mc H}\subset{\mc H}$ and $\PP(\tilde{\mc H})\subset\PP(\mc H)$. The respective varieties of coherent states are $\XX=\mathrm{Segre}(\PP(\mc H_1)\times...\times\PP(\mc H_L))\subset\PP(\mc H)$ and $\tilde{\XX}=\mathrm{Segre}(\PP(U_1)\times...\times\PP(U_L))\subset\PP(\tilde{\mc H})$, these are the varieties of decomposable tensors in the two systems. We have $\tilde{\XX}=\XX\cap\PP(\tilde{\mc H})$.\\

{\it Case 2: Bosons.} Let $\mc H_1$ be a finite dimensional vector space and ${\mc H}=S^L(\mc H_1)$. Let $U\subset \mc H_1$ be a subspace and $\tilde{\mc H}=S^L(U)$. We have natural embeddings $\tilde{\mc H}\subset{\mc H}$ and $\PP(\tilde{\mc H})\subset\PP(\mc H)$. The varieties of coherent states are $\XX=\mathrm{Veronese}_L(\PP(\mc H_1))\subset\PP(\mc H)$ and $\tilde{\XX}=\mathrm{Veronese}_L(\PP(U))\subset\PP(\tilde{\mc H})$. We have $\tilde{\XX}=\XX\cap\PP(\tilde{\mc H})$.\\

{\it Case 3: Fermions.} Let $\mc H_1$ be a finite dimensional complex vector space and ${\mc H}=\bigwedge^L \mc H_1$. Let $U\subset \mc H_1$ be a subspace with $\dim U\leq L$ and let $\tilde{\mc H}=\bigwedge^L U$. We have natural embeddings $\tilde{\mc H}\subset{\mc H}$ and $\PP(\tilde{\mc H})\subset\PP(\mc H)$. The varieties of coherent states are the Pl\"ucker embeddings of the Grassmann varieties $\XX=\mathrm{Pl}(Gr(L,\mc H_1))\subset\PP(\mc H)$ and $\tilde{\XX}=\mathrm{Pl}(Gr(L,U))\subset\PP(\tilde{\mc H})$. We have $\tilde{\XX}=\XX\cap\PP(\tilde{\mc H})$.\\

\begin{lemma}\label{Lemma Reduction}
Suppose we are in one of the three situations described above. Let $[\psi]\in\PP(\tilde{\mc H})$. Then
$$
{\rm rk}_\XX[\psi] = {\rm rk}_{\tilde{\XX}}[\psi] \;.
$$
\end{lemma}

\begin{proof} Since the proofs in the three cases are completely analogous, we shall only treat the case of bosons. Denote $r={\rm rk}_\XX[\psi]$ and let
$$
\psi = x_1 + ... + x_r \;, \quad [x_l]\in\XX \;,
$$
be a minimal expression for $\psi$. We have $x_l=y_l^L$ for some $y_l\in \mc H_1$. We want to show that the $[x_l]$'s may actually be chosen in $\tilde{\XX}$, which means that the $y_l$'s can be chosen in $U$.

Let $u_1,...,u_m,v_{1},...,v_{n-m}$ be a basis of $\mc H_1$ such that $u_1,...,u_m$ is a basis of $U$. We can write the $y_l$ in this basis, say
$$
y_l = \sum\limits_{i=1}^{m} a_l^i u_i + \sum\limits_{j=1}^{n-m} b_l^j v_j \;.
$$
We can use these expressions to write $\psi$ in terms of the monomial basis of $S^L(\mc H_1)$. Since $\psi\in S^L(U)$, the final expression for $\psi$ will have nonzero coefficients only in front of the monomials in the $u_j$'s, no $v_j$'s will be included. But now, observe that one gets the same result if instead of $y_l$ one takes
$$
z_l = \sum\limits_{i=1}^{m} a_j^i u_i \;.
$$
So, we have
$$
z_1^k+...+z_r^k = y_1^k+...y_r^k = \psi \;.
$$
Setting $x_l'=z_l^k$, we get $[x_l']\in\tilde{\XX}$ and $\psi=x_1'+...+x_r'$ as desired.
\end{proof}

\subsubsection{Basic case for distinguishable particles: $\CC^2\otimes\CC^2\otimes\CC^2$}\label{Sec Example 3 qubits}

Here we construct an explicit approximation by states of rank 2 for the $W$-state of a system of three qubits. Let $v_1,v_2$ be a basis of $\CC^2$. In fact, the state which is local unitary equivalent to $W$ (so of the same rank as $W$) is obtained as a limit point of the action of some one-parameter subgroup of $GL_2(\CC)$. The one-parameter subgroup is
$$
A(a) =\frac{1}{2} \begin{pmatrix} a+a^{-1} & a-a^{-1} \\ a-a^{-1} & a+a^{-1} \end{pmatrix}\,,\quad a\in\CC^\times \;,
$$
which may also be written as
$$
A(a) = g_0 A_1(a) g_0^{-1} ,\; \textrm{where} \;\; A_1(a)=\begin{pmatrix} a  & 0 \\ 0 & a^{-1} \end{pmatrix}\in GL_2(\CC) \;,\; g_0=\frac{1}{\sqrt 2}\begin{pmatrix} 1  & -1 \\ 1 & 1 \end{pmatrix}\in SU(2) \;.
$$
A direct calculation shows that we have the following convergence in $\PP(\CC^2\otimes\CC^2\otimes\CC^2)$.
\begin{gather*}
A(a)^{\otimes 3}[v_1\otimes v_1\otimes v_1 + v_2\otimes v_2\otimes v_2] \stackrel{a\rw 0}{\lw} g_0^{\otimes 3}[v_1\otimes v_2\otimes v_2 + v_2\otimes v_1\otimes v_2 + v_2\otimes v_2\otimes v_1]\,.
\end{gather*}
So the state $g_0W$, and hence $W$, can be approximated by states of rank $2$. On the other hand, it can be checked by direct calculation that ${\rm rk}[W]=3$. Thus $W$ is an exceptional state. For an argument which uses orbit structure see Remark \ref{w-analogous}.

\subsubsection{Bosons: $S^L(\CC^{n})$}\label{Sec Example Bosons}

Here we show that if $L\geq3$, then the system $S^L(\CC^n)$ of $L$ bosons has exceptional states.

Let $v_1,...,v_n$ be a basis of $\CC^n$, so that the monomials of degree $L$ in the $v_j$ form a basis of ${\mc H}=S^L\CC^n$. Let $\XX=\mathrm{Ver}_L(\PP^{n-1})\subset\PP(\mc H)$. Let
\begin{equation}
\psi=v_1^L+(-1)^{L+1}v_n^L \;.
\end{equation}
Consider the one-parameter subgroup of $G=GL_n(\CC)$ given by
$$
A(a) = \begin{pmatrix} \frac{a+a^{-1}}{2} & 0 & \frac{a-a^{-1}}{2} \\
                       0 & \II_{n-2} & 0 \\
                       \frac{a-a^{-1}}{2} & 0 & \frac{a+a^{-1}}{2}
       \end{pmatrix}\,,\quad
       a\in\CC^\times \,,
$$
We have $A(a)=g_0A_1(a)g_0^{-1}$, where
$$
A_1(a) = \begin{pmatrix} a & 0 & 0 \\
                       0 & \II_{n-2} & 0 \\
                       0 & 0 & a^{-1}
       \end{pmatrix} \in SL_n\CC \quad,\quad
g_0 = \begin{pmatrix} \frac{1}{\sqrt{2}} & 0 & -\frac{1}{\sqrt{2}} \\
                       0 & \II_{n-2} & 0 \\
                       \frac{1}{\sqrt{2}} & 0 & \frac{1}{\sqrt{2}}
       \end{pmatrix}.
$$
A simple application of the binomial formula shows that
$$
A(a)[\psi] \stackrel{a\rw 0}{\lw} g_0 [v_1v_n^{L-1}] \quad {\rm in} \quad \PP(\mc H) \;.
$$
Hence
$$
\ul{\rm rk}[v_1v_n^{L-1}] = 2 \;.
$$
But, according to \cite{Coma-Seig-2011} or \cite{CarCatGer-2012-Waring} (the relevant result is stated here as Theorem \ref{Theo rank of monom}), we have
$$
{\rm rk}[v_1v_n^{L-1}] = L \;.
$$
Thus
$$
L>2 \quad \Lw \quad {\rm rk} \ne \ul{\rm rk} \;.
$$
Finally, note that $[v_1v_n^{L-1}]$ can be regarded as a $W$-type state of $L$ bosons, i.e. a symmetric state in which $L-1$ out of $L$ bosons are in the highest excited state.

\subsubsection{Basic case for fermions: $\bigwedge^3\CC^6$}\label{Sec Example 3 fermions}

Let $G=SL_6(\CC)$, ${\mc H}=\bigwedge^3\CC^6$, $\XX=\mathrm{Pl}(Gr_3(\CC^6))\subset\PP(\mc H)$. We shall show that $\PP(\mc H)$ contains exceptional states with respect to $\XX$.

Let $\PP=\PP(\mc H)$. We have $\dim\XX=9$ and $\dim\PP=19=2\dim\XX+1$. The ring of invariants is isomorphic to a polynomial ring in one variable, $\CC[\mc H]^G=\CC[F]$ and the generator $F$ is a polynomial with $\deg F=4$ (cf. \cite{Kac-1980} Table II). Let $Z\subset\PP$ be the quartic hypersurface defined by the vanishing of $F$. Then $Z$ is a $G$-invariant variety and $\XX\subset Z\subset\PP$.

Let ${\mc N}_\alpha=\{\psi\in{\mc H}:F(\psi)=\alpha\}$, for $\alpha\ne 0$. Proposition 3.3 in \cite{Kac-1980} states that ${\mc N}_\alpha$ is a single $G$-orbit. It follows that $G$ acts transitively on $\PP\setminus Z$. Hence the latter set consists exactly of the generic states of rank 2. It remains to understand the orbit structure in $Z$.

We consider now the stats of rank 2 in ${\mc H}$, and show that there are two types of such states. Any such state can be written as
$$
\psi = v_1\wedge v_2\wedge v_3 + v_4\wedge v_5\wedge v_6 \;,
$$
with some $v_j\in{\mc H}$. The first possibility is that $v_1,...,v_6$ form a basis of $\CC^6$. This is indeed the generic situation. If suitable Borel and Cartan subgroups of $SL_6$ are chosen, the two summands of $\psi$ are, respectively, the highest and lowest weight vectors in ${\mc H}$. The group $GL_6$ acts transitively on the set of all bases of $\CC^6$; the group $SL_6$ acts transitively on the set of their projective images. Thus the states of the first type form a single $G$-orbit $\XX_2'=\PP\setminus Z$. The second possibility is to have
$$
\dim (\rm{span}\{v_1,v_2,v_3\}\cap \rm{span}\{v_4,v_5,v_6\}) = 1 \;.
$$
If this is the case, by changing the vectors if necessary, we may reduce to the situation where $v_1=v_4$ and
$$
\psi= v_1\wedge (v_2\wedge v_3 + v_5\wedge v_6) \;,\quad {\rm with} \quad \rm{span}\{v_2,v_3\}\cap \rm{span}\{v_5,v_6\} = 0 \;.
$$
Since $v_2\wedge v_3 + v_5\wedge v_6$ has rank 2 in $\bigwedge^2\CC^6$ (with respect to $Gr(2,\CC^6)$), we deduce that $\psi$ has indeed rank 2 in ${\mc H}$. The state $\phi$ is not of the first type described above, because the action of $GL_6$ respects linear dependencies. On the other hand, it is also clear that $GL_6$ acts transitively on the set $\XX_2''$ of states of this second type, and hence $SL_6$ acts transitively on the set of their images in $\PP$. Note that if
$$
\dim (\rm{span}\{v_1,v_2,v_3\}\cap \rm{span}\{v_4,v_5,v_6\}) > 1 \;,
$$
then ${\rm rk}[\psi]=1$. We can conclude that there are exactly two $G$-orbits consisting of states of rank 2, namely
$$
\XX_2 = \XX_2' \sqcup \XX_2'',\quad \XX_2'=\PP\setminus Z, \quad \XX_2''=Z\cap\XX_2 \,.
$$
Now, we verify the presence of states of rank 3. Set
$$
\phi = v_{1,1}\wedge v_{1,2}\wedge v_{1,3} + v_{2,1}\wedge v_{2,2}\wedge v_{2,3} + v_{3,1}\wedge v_{3,2}\wedge v_{3,3} \;,
$$
with $v_{j,k}\in{\mc H}$. Denote $V_j=\rm{span}\{v_{j,1},v_{j,2},v_{j,3}\}$ for $j=1,2,3$. If $\dim(V_i\cap V_j)>1$ for some pair of distinct indices, then the corresponding two summands of $\phi$ can be joined into a single simple tensor, and hence $\phi$ has rank at most 2. If $\dim(V_i\cap V_j)=0$ for some pair of distinct indices, say $\dim(V_1\cap V_2)=0$, then either $\dim(V_1\cap V_3)>1$ or $\dim(V_2\cap V_3)>1$, and we are brought to the previous case. Suppose $\dim(V_i\cap V_j)=1$ for all pairs of distinct indices. Here is an example of such a state:
$$
\phi = v_1\wedge v_2\wedge v_4 - v_1\wedge v_3\wedge v_5 + v_2\wedge v_3\wedge v_6 \;,
$$
where $\{v_1,...,v_6\}$ is a basis of $\CC^6$. To show that this state has indeed rank 3, we propose the following argument. If $\phi$ had rank 2, then it must belong either to $\XX_2'$ or $\XX_2''$ according to the above construction. It does not belong to $\XX_2''$, because $V_1+V_2+V_3=\CC^6$, while the components of any state from $\XX_2''$ span a 5-dimensional subspace of $\CC^6$. For an analogous reason $[\phi]$ does not belong to $\XX$. To show that $\phi$ does not belong to $\XX_2'$, we simply observe that $\phi$ is the limit point of an orbit of a one-parameter subgroup of $G$, through a point in $\XX_2'$. Take
$$
[\psi] = [v_1\wedge v_2\wedge v_3 + v_4\wedge v_5\wedge v_6] \in \XX_2' \;.
$$
Consider the one-parameter subgroup of $GL_6$ given by
$$
A(a) = \frac{1}{8}\begin{pmatrix} a+a^{-1} & 0 & 0 & 0 & 0 & a-a^{-1} \\
                       0 & a+a^{-1} & 0 & 0 & a-a^{-1} & 0 \\
                       0 & 0 & a+a^{-1} & a-a^{-1} & 0 & 0 \\
                       0 & 0 & a-a^{-1} & a+a^{-1} & 0 & 0 \\
                       0 & a-a^{-1} & 0 & 0 & a+a^{-1} & 0 \\
                       a-a^{-1} & 0 & 0 & 0 & 0 & a+a^{-1}
       \end{pmatrix}\,,\quad
       a\in\CC^\times \,,
$$
It can also be written as $A(a)=g_0A_1(a)g_0^{-1}$, with the following $A_1(a)\in SL_6\CC$ and $g_0\in SU_6$:
$$
A_1(a) = \begin{pmatrix} a & 0 & 0 & 0 & 0 & 0 \\
                       0 & a & 0 & 0 & 0 & 0 \\
                       0 & 0 & a & 0 & 0 & 0 \\
                       0 & 0 & 0 & a^{-1} & 0 & 0 \\
                       0 & 0 & 0 & 0 & a^{-1} & 0 \\
                       0 & 0 & 0 & 0 & 0 & a^{-1}
       \end{pmatrix} \quad,\quad
g_0 = \frac{1}{\sqrt 8}\begin{pmatrix} 1 & 0 & 0 & 0 & 0 & -1 \\
                       0 & 1 & 0 & 0 & -1 & 0 \\
                       0 & 0 & 1 & -1 & 0 & 0 \\
                       0 & 0 & 1 & 1 & 0 & 0 \\
                       0 & 1 & 0 & 0 & 1 & 0 \\
                       1 & 0 & 0 & 0 & 0 & 1
       \end{pmatrix}\;.
$$
A direct calculation shows that
$$
A(a)[\psi] \stackrel{a\rw 0}{\lw} g_0[v_1\wedge v_2\wedge v_4 - v_1\wedge v_3\wedge v_5 + v_2\wedge v_3\wedge v_6]=g_0[\phi] \quad {\rm in} \quad \PP(\mc H) \;.
$$
Thus the state $[g_0\phi]$ is a limit point of the orbit of $A(a)$ through $[\psi]$. We want to show that $[g_0\phi]\notin \XX_2'$, which would imply that $[\phi]\notin\XX_2'$. To this end notice that both $\psi$ and $\phi$ are critical points of the norm of the momentum map which is given by $||\mu([v])||^2=\mathrm{tr}(\mu([v])^2)$. However, the spectra of their reduced density matrices are different. This means $\psi$ and $\phi$ are not local unitary equivalent, i.e. do not belong to the same $K$-orbit. It is known \cite{Ness84} that any $G$-orbit can contain at most one $K$-orbit with critical points of the norm of momentum map \cite{Ness84} (see also \cite{SOK12} for a discussion of $||\mu([v])||^2$ properties in the entanglement setting). Hence $\phi$ and $\psi$ belong to different $G$-orbits and $[\phi]\notin\XX_2'$.


To summarize, we have $[\phi]\notin\XX_2\cup\XX$. We can conclude that
$$
{\rm rk}[\phi]=3 \quad {\rm and} \quad \ul{\rm rk}[\phi] = 2 \;,
$$
which means that $\phi$ is an exceptional state. Note that the sate $\phi$ is a fermionic counterpart of the three qubit $W$ state.
\begin{rem}\label{w-analogous}
Notice that the same reasoning can be mutatis mutandis applied to the three qubit case to prove that ${\rm rk}[W]=3$. We have exactly one invariant, the so-called hyperdeterminant. Moreover, the action of $SL_2(\CC)^{\times 3}$ on $\PP(\CC^2\otimes \CC^2\otimes \CC^2)$ has a dense open orbit. It can be shown that all states of rank $2$ in $\CC^2\otimes \CC^2\otimes \CC^2$ belong to either the open orbit or one of the three orbits of bi-separable states. Clearly $W$ is not bi-separable. What is left is to show that it does not belong to the open orbit. To this end, once again, notice that $W$ and $\psi=v_1\otimes v_1\otimes v_1+ v_2\otimes v_2\otimes v_2$ are critical points of the norm of the momentum map. However, the spectra of their reduced density matrices are different. This means $W$ and $\phi$ are not local unitary equivalent, i.e. do not belong to the same $K$-orbit. But any $G$-orbit can contain at most one $K$-orbit with critical points of the norm of momentum map \cite{Ness84}. Hence $W$ and $\
psi$ belong to
different $G$-orbits and ${\rm rk}[W]=3$.
\end{rem}

\section{The exceptional states for systems with known $G$-orbit structure}

So far, we have established the existence of the exceptional states stemming from the second secant variety, $\sigma_2$. In this section we give a list of all exceptional states which appear for distinguishable particles when: (a) the number of $G$-orbits in $\mathbb{P}(\mathcal{H})$ is finite and (b) four qubits system, where the number of $G$ orbits is infinite but their structure is known. The review is based on \cite{CD07,Bucz-Lands-2011}.

\subsection{$\mathbb{P}(\mathcal{H})$ with finite number of G-orbits}

In this subsection we consider all systems of distinguishable particles for which the number of $G$-orbits in $\mathbb{P}(\mathcal{H})$ is finite. Note that since for $\mathcal{H}=(\mathbb{C}^2)^{\otimes 4}$ the number of $G$-orbits in $\mathbb{P}(\mathcal{H})$ is already infinite \cite{CD07} we can restrict to $3$-partite systems. The requirement of finite number of $G$-orbits reduces our considerations to two possibilities \cite{Parfenov-2001}
\begin{enumerate}
\item $\mathcal{H}=\mathbb{C}^2\otimes\mathbb{C}^2\otimes\mathbb{C}^N$, where $N\geq 2$, $G=SL_2(\CC)\times SL_2(\CC)\times SL_N(\CC)$
\item $\mathcal{H}=\mathbb{C}^2\otimes\mathbb{C}^3\otimes\mathbb{C}^N$, where $N\geq 3$, $G=SL_2(\CC)\times SL_3(\CC)\times SL_N(\CC)$
\end{enumerate}

\paragraph{Case 1}
As we already know, when $N=2$ there are six $G$-orbits. A complete description of the orbit structure for $N>2$ can be found in \cite{Bucz-Lands-2011} (see table 1). Interestingly, the number of orbits stabilizes for $N\geq4$. More precisely, when $N=3$ there are two additional $G$-orbits compared with $N=2$, i.e. altogether there are eight $G$-orbits. Starting from $N=4$, the number of $G$-orbits is nine and does not change with $N$. The rank and border rank of states belonging to each $G$-orbit have been calculated and can be found in table 1 of \cite{Bucz-Lands-2011}.  There is only one $G$-orbit containing exceptional states: the $G$-orbit through $W=\ket{100}+\ket{010}+\ket{001}$. We have ${\rm rk}[W]=3$ and \ul{\rm rk}W = 2.
\paragraph{Case 2}
For $\mathcal{H}=\mathbb{C}^2\otimes\mathbb{C}^3\otimes\mathbb{C}^N$ the orbit structure is much richer. The number of orbits of $G=SL_2(\mathbb{C})\times SL_3(\mathbb{C})\times SL_N(\mathbb{\CC})$ stabilizes for $N\geq 6$. More precisely

\begin{itemize}
	\item $N=3$: 18 $G$-orbits
	\item $N=4$: 23 $G$-orbits
	\item $N=5$: 25 $G$-orbits
	\item $N\geq 6$: 26 $G$-orbits
\end{itemize}
The complete list of these orbits together with their dimensions, exemplary states belonging to each one, ranks and border ranks can be found in tables 2, 3 and 4 of \cite{Bucz-Lands-2011}. Here we list the exceptional states, their rank, border rank and the smallest $N$ for which they appear

\begin{enumerate}
	\item $\Psi_1=\ket{100}+\ket{010}+\ket{001}$, ${\rm rk}\Psi_1=3$, $\ul{\rm rk}\Psi_1=2$, $N=2$,
	\item $\Psi_2=\ket{0}\otimes(\ket{00}+\ket{11})+\ket{1}\otimes (\ket{01}+\ket{22})$, ${\rm rk}\Psi_2=4$, $\ul{\rm rk}\Psi_2=3$, $N=3$,
	\item $\Psi_3=\ket{0}\otimes(\ket{00}+\ket{11}+\ket{22})+\ket{1}\otimes (\ket{01}+\ket{12})$, ${\rm rk}\Psi_3=4$, $\ul{\rm rk}\Psi_3=3$, $N=3$,
	\item $\Psi_4=\ket{0}\otimes(\ket{00}+\ket{11}+\ket{22})+\ket{101}$, ${\rm rk}\Psi_4=4$, $\ul{\rm rk}\Psi_4=3$, $N=3$,
	\item $\Psi_5=\ket{0}\otimes(\ket{00}+\ket{12})+\ket{1}\otimes (\ket{01}+\ket{22})$, ${\rm rk}\Psi_5=4$, $\ul{\rm rk}\Psi_5=3$, $N=3$,
	\item $\Psi_6=\ket{0}\otimes(\ket{00}+\ket{12}+\ket{23})+\ket{1}\otimes (\ket{01}+\ket{13})$, ${\rm rk}\Psi_6=5$, $\ul{\rm rk}\Psi_6=4$, $N=4$.
\end{enumerate}

Interestingly, these states correspond to nontrivial entanglement classes found recently  by means of the so-called sub-Schmidt decomposition by the authors of \cite{CT06}.

\subsection{Infinite number of orbits - four qubits}
The number of $G=SL_2(\CC)^{\times 4}$-orbits in $\mathbb{P}(\mathcal{H})$ where $\mathcal{H}=(\mathbb{C}^{2})^{\otimes 4}$ is infinite. Nevertheless, the orbit structure is explicitly known \cite{CD07,kunraty}. Note that, according to theorem \ref{Coro SecVar Cubits}, the minimal number of simple tensors in $(\CC^2)^{\otimes 4}$ necessary to express a generic tensor as a linear combination is the expected number $r_g= \left\lceil \frac{2^4}{4+1} \right\rceil=4$. It was proven by Bryli\'nski \cite{Brylinski02} that the maximal tensor rank of a four qubits state is $4$. Therefore, tensors of rank $4$ cannot approximate tensors of higher rank and the border rank of any exceptional $[\psi]\in \mathbb{P}(\mathcal{H})$ can be at most $3$. Consequently, the exceptional states can belong only to $\sigma_2$ and $\sigma_3$, i.e. they arise from the closure of states of rank $2$ and rank $3$. These states have been completely determined in \cite{CD07}.

\paragraph{States of rank 2 and $\sigma_2$}
By Proposition 5.1 of \cite{CD07} the states of rank $2$ belong to three $G$-orbits, with representatives:
\begin{enumerate}
  \item $\Psi_1=[\ket{0000}+\ket{1111}]$
  \item $\Psi_2=[\ket{0}\otimes(\ket{000}+\ket{111})]$
  \item $\Psi_3=[\ket{00}\otimes (\ket{00}+\ket{11})]$
\end{enumerate}
Therefore, the closures of $G.[\Psi_{1,2,3}]$ are the only source of exceptional states belonging $\sigma_2$. It is easy to see that $\overline{G.[\Psi_3]}\setminus G.[\Psi_3]$ contains only separable states. In $\overline{G.[\Psi_2]}$ we find the state $\phi_2=\ket{0}\otimes(\ket{001}+\ket{010}+\ket{100})$ which has rank $3$ and border rank $2$ and hence is exceptional. Note that this state is a tensor product of $\ket{0}$ and a three qubit $W$ state. The last exceptional state type stems for  $\overline{G.[\Psi_1]}$ and is a four qubit $W$ state, i.e. $[\ket{0001}+\ket{0010}+\ket{0100}+\ket{1000}]$.

\paragraph{States of rank 3 and $\sigma_3$}
The second and last source of exceptional states is $\sigma_3$. It was shown by Bryli\'nski \cite{Brylinski02} that $\sigma_3$ is an irreducible algebraic variety given as the zero locus of two out of the four generating invariant polynomials of the $G$-action on $\mathcal{H}$ (these polynomials are denoted by $L$ and $M$ in \cite{CD07}). Since the $G$-orbits in $\mathcal{H}$ are known explicitly \cite{CD07} one can single out those on which the above mentioned polynomials vanish. If in addition one knows representatives of the $G$-orbits of rank $2$ and  $3$, one can determine the exceptional states stemming from $\sigma_3$. Following this reasoning the authors of \cite{CD07} found that the closure of the set of states of rank at most three contains only one $G$-orbit of states of rank $4$. This is $G$-orbit of the four qubit $W$-state $[\ket{0001}+\ket{0010}+\ket{0100}+\ket{1000}]$.

The conlusion is that, for four qubits the only exceptional states are $G$-orbits through $\ket\otimes(\ket{001}+\ket{010}+\ket{100})$ and $\ket{0001}+\ket{0010}+\ket{0100}+\ket{1000}]$ which are three and four qubit $W$-states.

\section{Summary}
We have established a connection between spherical actions of a reductive groups and the border rank of a state - a typical notion in Secant varieties theory. More precisely we showed that sphericity implies that states of a given rank cannot be approximated by states of a lower rank, i.e. there are no exceptional states. For three, important form quantum entanglement perspective, cases of distinguishable, fermionic and bosonic particles we also show that non-sphericity implies the existence of the exceptional states. We showed that the corresponding exceptional states belong to non-bipartite entanglement classes of $W$-type. Finally, we concluded that the existence of the exceptional states is a state-type obstruction for deciding the local unitary equivalence of states by means of the one-particle reduced density matrices.

A desired result would be a classification of exceptional states for the considered three types of systems, that is, L bosons, fermions and distinguishable particles. Note that even in case of $L$ qubits such a classification is not known. All examples of exceptional states in $L$ qubit systems, known to us, are of $W$ type. It would interesting to find out if there are exceptional states of different types. In order to find new examples one should look at the system of at least five qubits. It is because ${\rm rk}W=L$ and the generic rank $\left\lceil \frac{2^L}{L+1} \right\rceil>L$, when $L\geq 5$. For example when $L=6$ the generic rank is $10$ and rank of $W$ is 6. It would be very surprising if the sequences of tensors of rank $6,7,8,9$ or $10$ do not give tensors of higher rank. In our opinion the fact that for four qubits the maximal rank and generic rank overlap, which greatly simplifies the problem, might be a specific low dimensional phenomenon. Note also that all exceptional states found in
the considered examples are unstable states in the sense of geometric invariant theory (GIT). This in turn means that they cannot be taken by SLOCC operations (G-action) to a state with maximally mixed reduced one-particle density matrices. The mutual relationship between classification of states with respect to rank or border rank and with respect to GIT stability seems to be very interesting problem from both physics and mathematics points of view which we intend to follow.

\section*{Acknowledgments}

We would like to thank Marek Ku\'s for useful discussions and critical reading of the manuscript and Alan Huckleberry for his support and interest in our work. AS would like to thank Matthias Christandl, Marek Ku\'s, Michal Oszmaniec and Michael Walter for inspiring discussions concerning exceptional states and GIT during his visit in ETH Zurich. We also gratefully acknowledge the support of SFB/TR12 Symmetries and Universality in Mesoscopic Systems program of the Deutsche Forschungsgemeischaft, and Polish MNiSW Iuventus Plus grant no. IP2011048471.

\bibliographystyle{plain}

\end{document}